\documentclass[11pt,reqno]{amsart}

\usepackage{amscd,amssymb,amsmath,amsthm}
\usepackage[arrow,matrix]{xy}
\usepackage{graphicx}
\usepackage{epstopdf}
\usepackage{color}
\newtheorem{thm}{Theorem}

\newtheorem{lemma}{Lemma}
\newtheorem{pro}{Proposition}
\newtheorem{rk}{Remark}

\numberwithin{equation}{section} 

\def\s{\sigma}

\def\s{\sigma}

\def\s{\sigma}

\begin{document}
\title[Thermodynamics of interacting system of DNAs]{
Thermodynamics of interacting system of DNAs}

\author{U. A. Rozikov}

\address{U.\ A.\ Rozikov, \\ Institute of mathematics,
81, Mirzo Ulyg'bek str., 100170, Tashkent, Uzbekistan.}
\email {rozikovu@yandex.ru}

\begin{abstract} We define a DNA as a sequence of $1, 2$'s and embed it on a path of Cayley tree in such a way
that each vertex of the Cayley tree belongs only to one of DNA and each DNA has its own countably
many set of neighboring DNAs. The Hamiltonian of this set of
DNAs is a model with two spin values considered as DNA base pairs.
We describe translation
invariant Gibbs measures (TIGM) of the model on the Cayley tree of order two and use them to study thermodynamic properties of the model of DNAs.
We show that there is a critical temperature $T_{\rm c}$ such that (i) if temperature $T\geq T_{\rm c}$ then
 there exists unique TIGM; (ii) if $T<T_{\rm c}$ then
 there are three TIGMs. Each TIGM gives a phase of the set of DNAs.   In case of very high and very low temperatures we give
 stationary distributions and typical configurations of the model.
\end{abstract}
\maketitle

{\bf Mathematics Subject Classifications (2010).} 92D20; 82B20; 60J10; 05C05.

{\bf{Key words.}} DNA, temperature, Cayley tree,
Gibbs measure.

\section{Introduction}

It is known that \cite{book} genetic information is carried in the linear sequence of nucleotides in DNA.
Each molecule of DNA is a double helix formed from two complementary strands of nucleotides
held together by hydrogen bonds between $G+C$ and $A+T$ base pairs, where  $C$=cytosine, $G$=guanine, $A$=adenine,
and $T$=thymine. Duplication of the genetic
information occurs by the use of one DNA strand as a template for formation of a complementary strand.
The genetic information stored in an organism's DNA contains the instructions for all the proteins
the organism will ever synthesize.

The structure of DNA can be described using
 methods of statistical physics (see \cite{Sw}, \cite{T}),
 where by a single DNA strand is modelled as a stochastic system
 of interacting bases with long-range correlations. This study makes an
 important connection between the structure of DNA sequence and {\it temperature};
 e.g., phase transitions in such a system may be interpreted as a conformational restructuring.

In the recent papers \cite{Rb} and \cite{Rp} we gave Ising and Potts models of DNAs, to study its thermodynamics.
Translation invariant Gibbs measures (TIGMs) of the set of these models of DNAs on the Cayley tree are studied.
 Note that non-uniqueness of Gibbs measure corresponds to
phase coexistence in the system of DNAs.  By properties of Markov chains (corresponding to TIGMs) Holliday junction and branches of DNAs are studied.

 In this paper we study thermodynamic properties of a new model of DNAs (see \cite{Rb} and \cite{Rp} for motivations of such investigations).

The paper is organized as follows.
In Section 2 we give main definitions from biology  and mathematics.
In Section 3, we give a system of functional equations, each solution of which
defines a consistent family of finite-dimensional Gibbs distributions and
guarantees existence of thermodynamic limit for such distributions.
We show that, depending on temperature, number of translation invariant Gibbs measures can be up to three.
 In the last subsection by properties of Markov chains (corresponding to Gibbs measures) we study interaction nature of DNAs. In case of very high and very low temperatures we give stationary distributions and typical configurations of the model.

 \section{Embedding DNAs on a tree}

Now following \cite{R} and \cite{Rb} we recall some definitions.

{\bf Cayley tree.} The Cayley tree $\Gamma^k$ of order $ k\geq 1 $ is an infinite tree,
i.e., a graph without cycles, such that exactly $k+1$ edges
originate from each vertex. Let $\Gamma^k=(V,L,i)$, where $V$ is the
set of vertices $\Gamma^k$, $L$ the set of edges and $i$ is the
incidence function setting each edge $l\in L$ into correspondence
with its endpoints $x, y \in V$. If $i (l) = \{ x, y \} $, then
the vertices $x$ and $y$ are called the {\it nearest neighbors},
denoted by $l = \langle x, y \rangle $. The distance $d(x,y), x, y
\in V$ on the Cayley tree is the number of edges of the shortest
path from $x$ to $y$:
$$
d (x, y) = \min\{d \,|\, \exists x=x_0, x_1,\dots, x_{d-1},
x_d=y\in V \ \ \mbox {such that} \ \ \langle x_0,
x_1\rangle,\dots, \langle x_{d-1}, x_d\rangle\} .$$

For a fixed $x^0\in V$ we set $ W_n = \ \{x\in V\ \ | \ \ d (x,
x^0) =n \}, $
\begin{equation}\label{p*}
 V_n = \ \{x\in V\ \ | \ \ d (x, x^0) \leq n \},\ \ L_n = \ \{l =
\langle x, y\rangle \in L \  | \ x, y \in V_n \}.
\end{equation}
For any $x\in V$ denote
$$
W_m(x)=\{y\in V: d(x,y)=m\}, \ \ m\geq 1.
$$
{\bf Group representation of the tree.}
Let $G_k$ be a free product of $k + 1$ cyclic groups of the
second order with generators $a_1, a_2,\dots, a_{k+1}$,
respectively, i.e. $a_i^2=e$, where $e$ is the unit element.

It is known that there exists a one-to-one correspondence between the set of vertices $V$ of the
Cayley tree $\Gamma^k$ and the group $G_k$ (see Chapter 1 of \cite{R} for properties of the group $G_k$).

We consider a normal subgroup $\mathcal H_0\subset G_k$ of infinite index constructed as follows.
  Let the mapping $\pi_0:\{a_1,...,a_{k+1}\}\longrightarrow \{e, a_1, a_2\}$ be defined by
    $$\pi_0(a_i)=\left\{%
\begin{array}{ll}
    a_i, & \hbox{if} \ \ i=1,2 \\
    e, & \hbox{if} \ \ i\ne 1,2. \\
\end{array}
\right.$$ Denote by $G_1$ the free product of cyclic groups
$\{e,a_1\}, \{e,a_2\}$. Consider
$$f_0(x)=f_0(a_{i_1}a_{i_2}...a_{i_m})=\pi_0(a_{i_1})
\pi_0(a_{i_2})\dots\pi_0(a_{i_m}).$$
    Then it is easy to see that $f_0$ is a homomorphism and hence
    $\mathcal H_0=\{x\in G_k: \ f_0(x)=e\}$ is a normal subgroup of
    infinite index.
Consider the factor group
$$G_k/\mathcal H_0=\{\mathcal H_0, \mathcal H_0(a_1), \mathcal H_0(a_2), \mathcal H_0(a_1a_2), \dots\},$$
where $\mathcal H_0(y)=\{x\in G_k: f_0(x)=y\}$. Denote
$$\mathcal H_n=\mathcal H_0(\underbrace{a_1a_2\dots}_n), \ \ \
\mathcal H_{-n}=\mathcal H_0(\underbrace{a_2a_1\dots}_n).$$
In this notation, the factor group can be represented as
$$ G_k/\mathcal H_0=\{\dots, \mathcal H_{-2}, \mathcal H_{-1}, \mathcal H_0, \mathcal H_1, \mathcal H_2, \dots\}.$$
We introduce the following equivalence relation on the set $G_k$: $x\sim y$ if $xy^{-1}\in \mathcal H_0$.
Then $G_k$ can be partitioned to countably many classes $\mathcal H_i$ of equivalent elements.
The partition of the Cayley tree $\Gamma^2$ w.r.t. $\mathcal H_0$ is shown in
Fig. \ref{fig9} (the elements of the class $\mathcal H_i$, $i\in \mathbb Z$, are merely denoted by $i$).

{\bf $\mathbb Z$-path.}
Denote
$$
q_i(x) = |W_1(x)\cap \mathcal H_i|, \ \ x\in G_k,
$$
where $|\cdot|$ is the counting measure of a set.
We note that (see \cite{RI}) if $x\in \mathcal H_m$, then
$$
q_{m-1}(x)=1,  \ \ q_m(x)=k-1, \ \ q_{m+1}(x)=1.
$$
From this fact it follows that
for any $x\in V$, if $x\in \mathcal H_m$ then there is a unique two-side-path (containing $x$) such that
the sequence of numbers of equivalence classes for vertices of this path
in one side are $m, m+1, m+2,\dots$ in the second side the sequence is $m, m-1,m-2,\dots$.
Thus the two-side-path has the sequence of numbers of equivalent classes as $\mathbb Z=\{...,-2,-1,0,1,2,...\}$.
Such a path is called $\mathbb Z$-path (In Fig. \ref{fig9} one can see the unique $\mathbb Z$-paths of each vertex of the tree.)

Since each vertex $x$ has its own $\mathbb Z$-path one can see that
the Cayley tree considered with respect to normal subgroup $\mathcal H_0$ contains infinitely many (countable) set of
$\mathbb Z$-paths.
\begin{figure}
   \includegraphics[width=13cm]{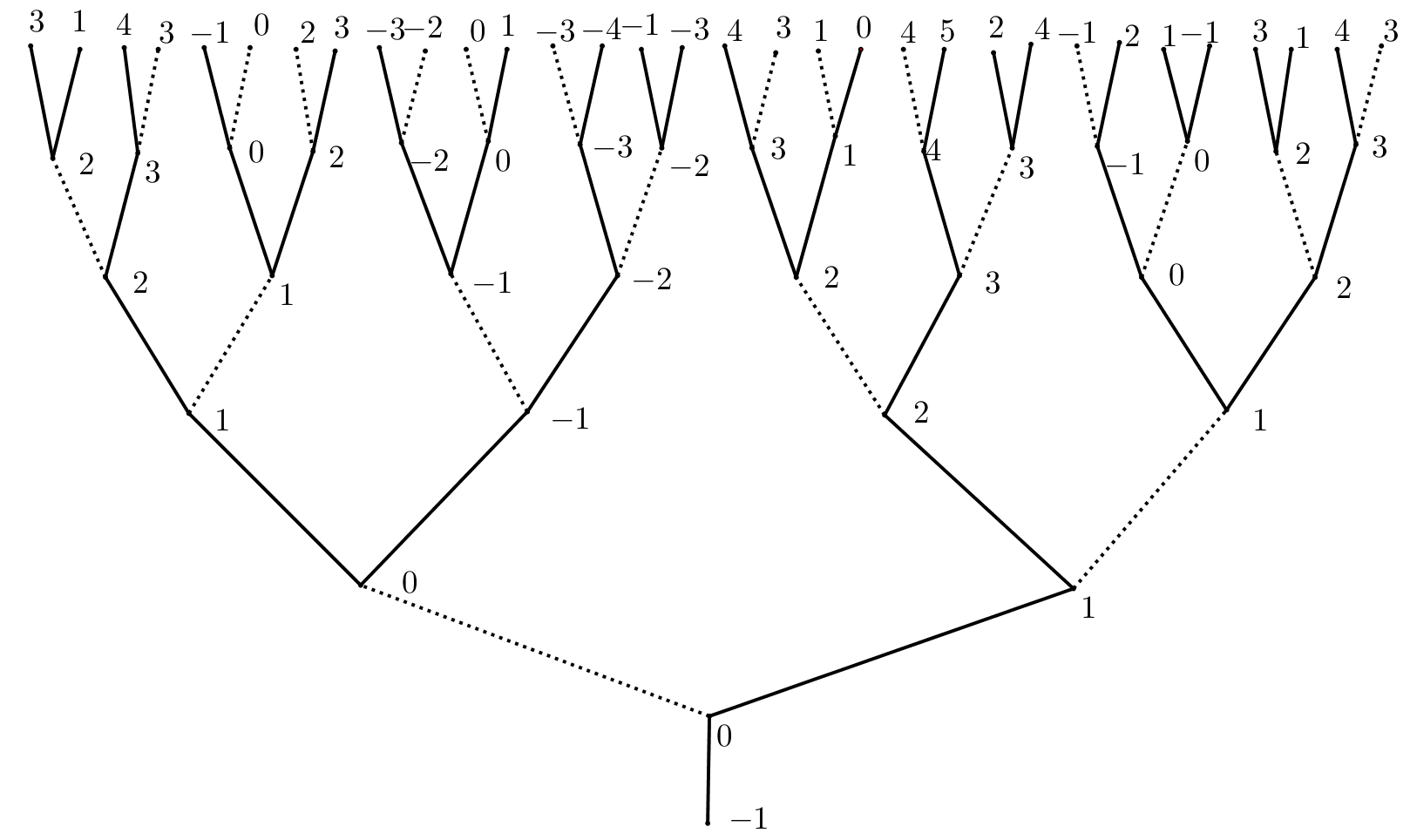}\\
  \caption{The partition of the Cayley tree $\Gamma^2$ w.r.t.
  $\mathcal H_0$, the elements of the class $\mathcal H_i$, $i\in \mathbb Z$,
  are denoted by $i$. $\mathbb Z$-pathes are solid lines.}\label{fig9}
\end{figure}

{\bf Tree-hierarchy of the set of DNAs.} Define a Cayley tree
hierarchy of the set of DNAs as follows.

    Given a configuration $\sigma$ on a Cayley
tree, since there are countably many $\mathbb Z$-paths we have a countably many distinct DNAs.
We say that two DNA {\it are neighbors} if there is an edge (of the Cayley tree) such that
one its endpoint belongs to the first DNA and
another endpoint of the edge belongs to the second DNA. By our
construction it is clear (see Fig. \ref{fig9}) that
such an edge is unique for each neighboring pair of DNAs. This edge
has equivalent endpoints, i.e. both endpoints belong to the same class
$\mathcal H_m$ for some $m\in \mathbb Z$.

Moreover these countably infinite set of DNAs have a hierarchy that:

(i) any two DNA do not intersect.

(ii) each DNA has its own countably many set of neighboring DNAs;

(iii) for any two neighboring DNAs, say $D_1$ and $D_2$, there exists a unique
edge $l=l(D_1,D_2)=\langle x,y\rangle$ with $x\sim y$ which connects DNAs;

(iv)  for any finite $n\geq 1$ the ball $V_n$ has intersection only with finitely many DNAs.

{\bf The model.} Let $V$ be the set of vertices of a Cayley tree.
Consider function $\sigma$ which assigns to each vertex $x\in V$,
values $\sigma(x)\in \{1, 2\}$. Here $1=A-T$ and $2=C-G$ are DNA base pairs.

A configuration $\sigma=\{\sigma(x),\ x\in V\}$ on vertices of the Cayley tree
is given by a function from $V$ to $\{1, 2\}$. The set of all
configurations in $V$ is denoted by $\Omega$. Configurations in
$V_n$  are defined analogously and the set of all
configurations in $V_n$ is denoted by $\Omega_n$.

The restriction of a configuration on
a $\mathbb Z$-path is called {\it a DNA}.

 We consider the following model of the energy of the configuration $\sigma$ of a set of DNAs:
\begin{equation}\label{h}
H(\sigma)=\sum_{\langle x,y\rangle\in L}f_{x,y}(\sigma(x),\sigma(y)),
\end{equation}
where

\begin{equation}\label{hd}f_{x,y}(\sigma(x),\sigma(y))=\left\{\begin{array}{ll}
f(\sigma(x),\sigma(y)), \ \ \mbox{if} \ \ \langle x,y\rangle\in \mathbb Z-path\\[2mm]
J\delta_{\sigma(x),\sigma(y)}, \ \ \mbox{if} \ \ \langle x,y\rangle\notin \mathbb Z-path
\end{array}
\right.
\end{equation}
$J>0$ is a coupling constant between neighboring DNAs, $\delta$ is the Kronecker delta, $\sigma(x)\in\{1, 2\}$ and
$\langle x, y\rangle$ stands for nearest neighbor vertices.
The function $f:\{1, 2\}^2\to \mathbb R$ gives interaction between DNA base pairs.

\section{Thermodynamics of the system of DNAs}

\subsection{System of functional equations of finite dimensional distributions}

Let $\Omega_n$ be the set of all
configurations on $V_n$.

Define a finite-dimensional distribution of a probability measure $\mu$ on $\Omega_n$ as
\begin{equation}\label{*}
\mu_n(\sigma_n)=Z_n^{-1}\exp\left\{\beta H_n(\sigma_n)+\sum_{y\in W_n}h_{\sigma(y), y}\right\},
\end{equation}
where $\beta=1/T$, $T>0$ is temperature,  $Z_n^{-1}$ is the normalizing factor,
$$\{h_{i,x}\in \mathbb R, i=1,2, \, x\in V\}$$ is a collection of real numbers and
$$H_n(\sigma_n)=\sum_{\langle x,y\rangle\in L_n}f_{x,y}(\sigma(x),\sigma(y)).$$

\begin{rk} We note that the quantities $\exp(h_{i,x})$ define a boundary
law in the sense of Definition 12.10 of \cite{Ge} (see also \cite{BR},  \cite{KRK} and \cite{KR}).
In our case these quantities mean a boundary law of our biological system of DNAs.
\end{rk}

We say that the probability distributions (\ref{*}) are compatible if for all
$n\geq 1$ and $\sigma_{n-1}\in \Omega_{n-1}$:
\begin{equation}\label{**}
\sum_{\omega_n\in \Omega_{W_n}}\mu_n(\sigma_{n-1}\vee \omega_n)=\mu_{n-1}(\sigma_{n-1}).
\end{equation}
Here $\sigma_{n-1}\vee \omega_n$ is the concatenation of the configurations.

In the formula (\ref{hd}) we take the function $f$ such that
$$f(1,1)=f(2,2), \ \ f(1,2)=f(2,1).$$
This assumption is natural in the structure of DNA.

For $x\in V_{n-1}$ denote
$$S(x)=\{t\in V_n: \langle x,t\rangle\}.$$

For $ x\in V$ we denote by $x_{\downarrow}$ the unique point of the set $\{y\in V:\langle x,y\rangle\}\setminus S(x)$.

It is easy to see that
$$S(x)\cap \mathbb Z-{\rm path}=\left\{\begin{array}{lll}
\{x_0, x_1\}\subset V, \ \ \mbox{if} \ \  \langle x_\downarrow, x\rangle\notin \mathbb Z-{\rm path}\\[2mm]
\{x_1\}\subset V, \ \ \ \ \ \ \mbox{if} \ \  \langle x_\downarrow, x\rangle\in \mathbb Z-{\rm path}
\end{array}
\right..$$
We denote
$$S_0(x)=S(x)\setminus \{x_0, x_1\}, \ \ \langle x_\downarrow, x\rangle\notin  \mathbb Z-{\rm path},$$
$$S_1(x)=S(x)\setminus \{x_1\}, \ \ \langle x_\downarrow, x\rangle\in  \mathbb Z-{\rm path}.$$

The following theorem can be proved as  Theorem 2.1 and Theorem 7.22 in \cite{R}.

\begin{thm}\label{ei} Probability distributions
$\mu_n(\sigma_n)$, $n=1,2,\ldots$, in
(\ref{*}) are compatible iff for any $x\in V\setminus {x^0}$
the following equations hold

  \begin{equation}\label{***}
  z_x={\hat z_{x_0}+b\over b \hat z_{x_0}+1}\cdot {\hat z_{x_1}+b\over b \hat z_{x_1}+1}
 \prod_{t\in S_0(x)}{z_{t}+\theta\over \theta z_{t}+1}, \ \
  \langle x_\downarrow, x\rangle\notin \mathbb Z-{\rm path},\end{equation}
  \begin{equation}\label{***a}
  \hat z_x={\hat z_{x_1}+b\over b \hat z_{x_1}+1}
 \prod_{t\in S_1(x)}{z_{t}+\theta\over \theta z_{t}+1}, \ \ \langle x_\downarrow, x\rangle\in \mathbb Z-{\rm path}.
 \end{equation}

Here,
\begin{equation}\label{alp}
\begin{array}{lll}
\theta=\exp\{\beta (f(1,2)-f(1,1))\}, \ \ b=\exp(-J\beta),\\[2mm]
z_{x}=\exp\left(h_{1,x}-h_{2,x}\right), \ \  \langle x_\downarrow, x\rangle\notin \mathbb Z-{\rm path},\\[2mm]
\hat z_{x}=\exp\left(h_{1,x}-h_{2,x}\right),  \ \ \langle x_\downarrow, x\rangle\in \mathbb Z-{\rm path}.
\end{array}
\end{equation}
\end{thm}

From Theorem \ref{ei} it follows that for any set of vectors
${\bf z}=\{(z_{x}, \hat z_t)\}$
satisfying the system of functional equations (\ref{***}) and  (\ref{***a}) there exists a unique Gibbs measure $\mu$ and vice versa. However,
the analysis of this system of nonlinear functional equations is not easy.
In next subsection we shall give several solutions to the system.

\begin{rk} The number of the solutions of equations (\ref{***}) and  (\ref{***a})  depends on  the
temperature and interaction parameters $b, \theta$. Note that if there is more than one solution to the system (\ref{***}) and  (\ref{***a}),
then there is more than one Gibbs measure (i.e. occurrence of
``phase'' transition for the model of DNAs).
\end{rk}

\subsection{Translation invariant Gibbs measures of the set of DNAs}

  In this subsection, we find solutions to the system of functional equations (\ref{***}) and  (\ref{***a}), which have the following form
  \begin{equation}\label{zzz}
z_{x}=u, \ \ \forall \langle x_\downarrow, x\rangle\notin {\mathbb Z}-{\rm path};
\ \ \hat z_{x}=w,\ \ \forall \langle x_\downarrow, x\rangle\in {\mathbb Z}-{\rm path}.
\end{equation}
  The Gibbs measure corresponding to such a solution is called {\it translation invariant}.

For $u, v$ from (\ref{***}) and (\ref{***a}) we get
\begin{equation}\label{f}
\begin{array}{ll}
u=\left({v+b\over bv+1}\right)^2\left({u+\theta\over \theta u+1}\right)^{k-2}\\[3mm]
v=\left({v+b\over bv+1}\right)\left({u+\theta\over \theta u+1}\right)^{k-1}.
\end{array}
\end{equation}
Here $u,v>0$.

It is clear that $u=v=1$ satisfies the system (\ref{f}) for any $k\geq 2$ and $b, \theta>0$.
In general, it is not easy to find all solution of this system.
Therefore for simplicity we consider the case $k=2$. Then from the system we get
$$u=\left({v+b\over bv+1}\right)^2,$$ and substituting it in the second equation (for $k=2$) we obtain
\begin{equation}\label{qn}
(v^2-1)\big(b(b^2+\theta)(v^2+1)+(b^2(3+\theta)+\theta-1)v\big)=0
\end{equation}
Since $v>0$ we have three solutions of (\ref{qn}): $ v_1=1$,
$$v_2={-(b^2(3+\theta)+\theta-1)-|b^2-1|\sqrt{(\theta-1)^2-4b^2}\over2b(b^2+\theta)},$$
$$v_3={-(b^2(3+\theta)+\theta-1)+|b^2-1|\sqrt{(\theta-1)^2-4b^2}\over2b(b^2+\theta)}.$$

The solutions $v_2$ and $v_3$ exist iff
$$\theta\in \left\{\begin{array}{ll}
(0, 1-2b]\cup [1+2b, +\infty) \ \ \mbox{if} \ \ 0<b<{1\over 2}\\[2mm]
 (1+2b, +\infty) \ \ \mbox{if} \ \ b\geq{1\over 2}.
 \end{array}
 \right.$$
Note that both $v_2$ and $v_3$ should be positive.
This can be  possible only if
\begin{equation}\label{v2} 0<b<{1\over 2} \ \ \mbox{and} \ \  \theta\leq 1-2b.
\end{equation}
Indeed, to see this rewrite equation
$$b(b^2+\theta)(v^2+1)+(b^2(3+\theta)+\theta-1)v=0$$ as
$$v+{1\over v}={1-\theta-b^2(3+\theta)\over b(b^2+\theta)}.$$
In case $v>0$ we should have
$${1-\theta-b^2(3+\theta)\over b(b^2+\theta)}\geq 2$$
which gives the condition (\ref{v2}). Note also that $v_2v_3=1$.

Introduce the following sets

$$U=(0, +\infty)\times (0, +\infty),$$
$$U_1=\{(b,\theta)\in U: b\leq 1/2, \ \ \theta< 1-2b\}.$$
Thus, for $k=2$ we proved the following
\begin{lemma}\label{l1} The following hold
\begin{itemize}
\item If $(b, \theta) \in U\setminus U_1$  then system (\ref{f}) has unique solution
$${\bf z_1}=(u_1, v_1)=(1,1);$$
\item If $(b, \theta) \in U_1$ then system (\ref{f}) has three solutions
$${\bf z_1}=(1,1), \ \  {\bf z_2}=(u_2,v_2), \ \ {\bf z_3}=(u_3,v_3),$$
where
$$u_i=\left({v_i+b\over bv_i+1}\right)^2, \ \ i=2,3.$$
\end{itemize}
\end{lemma}
Denote by $\mu_i$ the Gibbs measure which, by Theorem \ref{ei}, corresponds to the solution ${\bf z_i}$, $i=1,2,3.$

For given $f$ and $J$ the critical temperature $T_{\rm c}$ is defined by the equation $\theta+2b=1$, i.e.
 by formula (\ref{alp}) and $\beta=1/T$ we have
\begin{equation}\label{tc}
\theta+2b=\exp\left({1\over T}(f(1,2)-f(1,1))\right)+2\exp\left({-J\over T}\right)=1.
\end{equation}
This equation has solution iff
\begin{equation}\label{jf} J>0, \ \ f(1,2)-f(1,1)<0.
\end{equation}
Indeed, if $J\leq 0$ or $f(1,2)-f(1,1)\geq 0$ then LHS
of the equation (\ref{tc}) is strongly large than 1.
Moreover, in the case  $J>0$ and $f(1,2)-f(1,1)<0$ the function $t(T):=\theta+2b-1$ is an increasing function of $T$ with
value $-1$ at $T=0$ and value 2 at $T=+\infty$. Therefore, the equation
 has unique solution $T=T_{\rm c}$ (see Fig. \ref{fig2}).
\begin{figure}[h]
   \includegraphics[width=10cm]{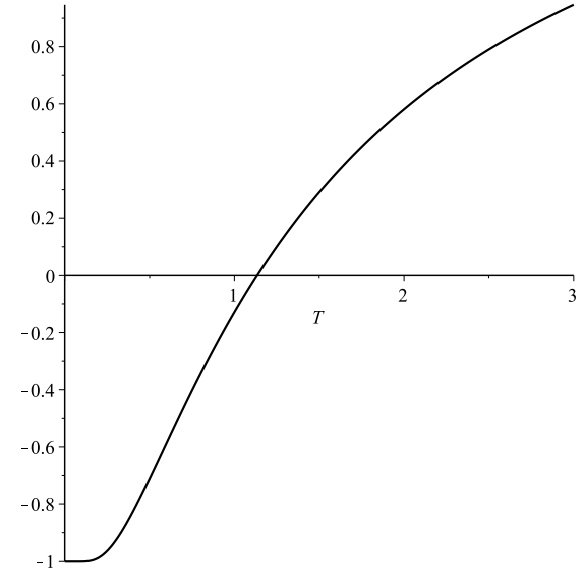}\\
  \caption{The graph of $t(T)$ for $J=1$, $f(1,2)-f(1,1)=-2$. In this case $T_c\approx 1.1346$.}\label{fig2}
\end{figure}
Summarizing above results we obtain the following

\begin{thm}\label{tii} For the model (\ref{h}) (given $J$ and $f$ with (\ref{jf}))
of DNAs on the Cayley tree of order $k=2$ the following statements are true
 \begin{itemize}
    \item[(1)] If the temperature $T\geq T_{\rm c}$ (where $T_{\rm c}\equiv T_{\rm c}(J, f)$
    the unique solution of  (\ref{tc})) then there is unique translation-invariant Gibbs measure $\mu_1$.

    \item[(2)] If $T< T_{\rm c}$ then there are 3 translation-invariant Gibbs measures $\mu_1, \mu_2, \mu_3$.
        \end{itemize}
 \end{thm}

 \begin{rk} 1) For given $J$ and $f$ with the conditions of Theorem \ref{tii}
 are related to conditions of Lemma \ref{l1} as following
 \begin{itemize}
    \item[(1)]  $T\geq T_{\rm c}$ iff $(b, \theta) \in U\setminus U_1$.

    \item[(2)]  $T< T_{\rm c}$ iff $(b, \theta) \in U_1$.
 \end{itemize}

 2) Analysis of system (\ref{f}) for the case $k\geq 3$ seems difficult. The case $k=2$ is already
 interesting enough to see biological interpretations of Theorem \ref{tii}.
 \end{rk}

 \subsection{Markov chains}
Recall
$$\begin{array}{ll}
z_{x}=\exp\left(h_{1,x}-h_{2,x}\right), \ \  \langle x_\downarrow, x\rangle\notin \mathbb Z-{\rm path},\\[2mm]
\hat z_{x}=\exp\left(h_{1,x}-h_{2,x}\right),  \ \ \langle x_\downarrow, x\rangle\in \mathbb Z-{\rm path}.
\end{array}$$
Using boundary law
$${\bf z}=\{(z_{x}, \hat z_t): \ \ \langle x_\downarrow, x\rangle\notin {\mathbb Z}-{\rm path};
\ \ \langle t_\downarrow, t\rangle\in {\mathbb Z}-{\rm path}\},$$
i.e., the solutions of system (\ref{***}), for marginals on the two-site volumes which consist of two adjacent sites $x,y$
we have
$$\mu(\s(x)=s,\s(y)=r)= \frac{1}{Z} \exp(h_{s,x}+\beta f_{x,y}(s, r)+h_{r,y}), \ \ s, r=1, 2,$$
where $Z$ is normalizing factor.

From this,  the relation between the boundary law and the transition matrix
for the associated tree-indexed Markov chain (Gibbs measure) is immediately obtained
from the formula of the conditional probability. The transition matrix of this Markov chain is defined as follows

$$\mathbb P^{\langle x, y \rangle}=\left( P_{ij}^{\langle x, y \rangle}\right)_{i,j=1,2}= \left\{\begin{array}{ll}
\left(\begin{array}{cc}
 {u\over u+\theta}& {\theta\over u+\theta}\\[2mm]
 {\theta u\over \theta u+1}&  {1\over \theta u+1}
\end{array}\right), \ \ \mbox{if} \ \  \langle x, y\rangle\in {\mathbb Z}-{\rm path}\\[4mm]
\left(\begin{array}{cc}
 {v\over v+b}& {b\over v+b}\\[2mm]
 {b v\over b v+1}& {1\over b v+1}
\end{array}\right), \ \ \mbox{if} \ \  \langle x, y\rangle\notin {\mathbb Z}-{\rm path},
\end{array}\right.
$$
where $(u,v)$ a solution of system (\ref{f}) (mentioned in Lemma \ref{l1}).

We note that each matrix $\mathbb P^{\langle x, y \rangle}$ does not depend on $\langle x,y\rangle$
itself but only depends on its relation with ${\mathbb Z}-$path.

It is easy to find the following stationary distributions
$$\pi^{\langle x, y \rangle}=\left\{\begin{array}{ll}
\left({u^2+\theta u\over u^2+(\theta+1)u+1}, {u+1\over u^2+(\theta+1)u+1}\right), \ \ \mbox{if} \ \  \langle x, y\rangle\in {\mathbb Z}-{\rm path}\\[3mm]
\left({v^2+b v\over v^2+(b+1)v+1}, {v+1\over v^2+(b+1)v+1}\right), \ \ \mbox{if} \ \  \langle x, y\rangle\notin {\mathbb Z}-{\rm path}
\end{array}\right.$$

The following is known (see p. 55 of \cite{Ge}) as ergodic theorem for positive
stochastic matrices.
\begin{thm}\label{to} Let $\mathbb P$ be a positive stochastic matrix and $\pi$
the unique probability vector with $\pi \mathbb P=\pi$ (i.e.  $\pi$ is stationary distribution). Then
$$\lim_{n\to \infty} x\mathbb P^n =\pi$$
for all initial vector $x$.
\end{thm}

In the case of non-uniqueness of Gibbs measure
(and corresponding Markov chains) we have different stationary states for
different measures. These depend on the temperature and on the fixed measure.

Recall that a DNA is a configuration $\sigma\in \{1,2\}^{\mathbb Z}$ on a $\mathbb Z$-path.
According to the definition of our model only neighboring DNAs may interact.
The interaction is through an edge $l=\langle x, y\rangle\notin \mathbb Z$-path connecting two DNAs when
configuration on this endpoints of the edge satisfy $\s(x)=\s(y)$. The neighboring DNAs do not interact if
$\sigma(x)\ne \sigma(y)$.

As a corollary of Theorem \ref{to} and above formulas of matrices and stationary distributions we obtain
the following
\begin{thm} In a stationary state of the set of DNAs we have
\begin{itemize}
\item[1.] Two neighboring DNAs do not
interact with the following probability, (denoted by $\mathbb P_i$ with respect to measure $\mu_i$, $i=1,2,3$):
$$\mathbb P_i={bv_i((b+1)v_i+2)\over (bv_i+1)(v_i^2+(b+1)v_i+1)},
$$
(Consequently, they do interact with probability $1-\mathbb P_i$)
where $(u_i,v_i)$ are defined in Lemma \ref{l1}.
\item[2.] Two neighboring base pairs (laying on vertices of
 an edge $l=\langle x, y\rangle\in \mathbb Z$-path)  in a DNA have distinct values (i.e.  $\sigma(x)\ne \sigma(y)$) with
 probability (denoted by $\mathbb Q_i$ with respect to measure $\mu_i$, $i=1,2,3$):
$$\mathbb Q_i={\theta u_i((\theta+1)u_i+2)\over (\theta u_i+1)(u_i^2+(\theta+1)u_i+1)},
$$
(Consequently, they have the same values with probability $1-\mathbb Q_i$)
\end{itemize}
\end{thm}
\begin{proof} Let $l=\langle x, y\rangle\notin \mathbb Z$-path be the edge connecting two neighboring DNAs.
By the definition of $\mathbb P_i$ we have
$$\mathbb P_i=\mu_i(\sigma(x)=1, \sigma(y)=2)+\mu_i(\sigma(x)=2, \sigma(y)=1).$$
Since $\mu_i$ generates a Markov chain we have
 $$\mu_i(\sigma(x)=a, \sigma(y)=c)=\pi^{\langle x, y \rangle}_aP^{\langle x, y \rangle}_{ac}.$$
This by using the above given formulas of $\pi^{\langle x, y \rangle}$ and $\mathbb P^{\langle x, y \rangle}$ completes the proof of part 1, the part 2 is similar.
\end{proof}

\begin{rk} Since each DNA has a countable set of neighbor DNAs, at the same
temperature, it my have interactions with several of its neighbors. In case when a DNA does not interact
with its neighbors then it is an isolated one. Interacting DNAs can be considered  as a
branched DNA. In case of coexistence of more than one Gibbs measures, branches of a DNA can consist
different phases and different stationary states.
\end{rk}

 Now we are interested to calculate
 the limit of  stationary distribution vectors $\pi^{\langle x, y \rangle,i}$,
 (which correspond to the Markov chain generated by the Gibbs measure $\mu_i$)
 in case when temperature $T\to 0$  and
when temperature $T\to +\infty$.
To calculate the limit observe that values  $u_{i}, v_i$, $i=1,2,3$ vary
with $T=1/\beta$.

Recall that measures $\mu_2$ and $\mu_3$ do not exist for $T>T_{\rm c}$.

\begin{pro} Independently on $\langle x, y \rangle$ the following equalities hold
\begin{itemize}
\item[-] The case of low temperature:
$$\lim_{T\to 0}\pi^{\langle x, y \rangle,1}=({1\over 3}, {2\over 3}),\ \
\lim_{T\to 0}\pi^{\langle x, y \rangle,2}=(0,1), \ \ \lim_{T\to 0}\pi^{\langle x, y \rangle,3}=(1,0).$$
\item[-] The case of high temperature:
$$\lim_{T\to +\infty}\pi^{\langle x, y \rangle,1}=\lim_{T\to T_{\rm c}}\pi^{\langle x, y \rangle,i}=({1\over 2},{1\over 2}), i=1,2,3.$$

\end{itemize}
\end{pro}
\begin{proof} Using explicit formulas for $u_{i}, v_i$, $i=1,2,3$
the proof consists simple calculations of limits.
\end{proof}

\begin{rk} Using this proposition one can give the structure of DNAs in low and high temperatures. Foe example,
in case $T\to 0$ the set of DNAs have the following stationary states (configurations):
\begin{itemize}
\item[Case $\mu_1$:] The base pairs $1=A-T, 2=C-G$s, in each point of a DNA can be seen with probability $1/3$ for 1 and $2/3$ for 2.
\item[Case $\mu_2$:] All DNAs are rigid and interact, i.e. $\sigma(x)=2$ for all $x\in \mathbb Z$-path.
\item[Case $\mu_3$:] All DNAs are rigid and interact, with $\sigma(x)=1$ for all $x\in \mathbb Z$-path.
\end{itemize}

  In case $T\to +\infty$ the sequence of $1, 2$s, in a DNA on the $\mathbb Z$-path, is free, with iid and equiprobable ($=1/2$), of $1$ and $2$s.
\end{rk}

\end{document}